\documentclass[a4paper,english]{lipics-v2016}

\usepackage{microtype}
\usepackage{etex}
\usepackage{booktabs}   
\usepackage{subcaption} 

\usepackage{microtype}
\usepackage{amsmath}
\usepackage{amssymb}
\usepackage{amsthm}
\usepackage[inline]{enumitem}
\usepackage{float}
\usepackage[T1]{fontenc}
\usepackage[utf8]{inputenc}
\usepackage{mathtools}
\usepackage{pgfplots}
\usepackage{tikz}
\usepackage{multicol}
\usepackage{color}
\usepackage{listings}
\usepackage{stmaryrd}

\definecolor{lipics}{rgb}{0.85,0.85,0.86}

\usetikzlibrary{arrows}
\usetikzlibrary{automata}
\usetikzlibrary{calc}
\usetikzlibrary{patterns}
\usetikzlibrary{backgrounds}
\makeatletter
\pgfarrowsdeclare{...}{...}{%
  \pgfutil@tempdima\pgflinewidth
  \pgfarrowsrightextend{+6.5\pgfutil@tempdima}%
  \pgfarrowsleftextend{+6.5\pgfutil@tempdima}%
}{
  \pgfutil@tempdima\pgflinewidth
  \pgfpathcircle{\pgfqpoint{+5\pgfutil@tempdima}{+0pt}}{+1.5\pgfutil@tempdima}%
  \pgfpathcircle{\pgfqpoint{+10\pgfutil@tempdima}{+0pt}}{+1.5\pgfutil@tempdima}%
  \pgfpathcircle{\pgfqpoint{+15\pgfutil@tempdima}{+0pt}}{+1.5\pgfutil@tempdima}%
  \pgfusepathqfill
}
\pgfarrowsdeclare{sss}{sss}{%
  \pgfutil@tempdima\pgflinewidth
  \pgfarrowsrightextend{+6\pgfutil@tempdima}%
  \pgfarrowsleftextend{+6\pgfutil@tempdima}%
}{
  \pgfutil@tempdima\pgflinewidth
  \pgfpathrectangle{\pgfqpoint{+2\pgfutil@tempdima}{+1.5\pgfutil@tempdima}}{\pgfqpoint{-\pgfutil@tempdima}{-\pgfutil@tempdima}}
  \pgfpathrectangle{\pgfqpoint{+4\pgfutil@tempdima}{+1.5\pgfutil@tempdima}}{\pgfqpoint{-\pgfutil@tempdima}{-\pgfutil@tempdima}}
  \pgfpathrectangle{\pgfqpoint{+6\pgfutil@tempdima}{+1.5\pgfutil@tempdima}}{\pgfqpoint{-\pgfutil@tempdima}{-\pgfutil@tempdima}}
  \pgfusepathqfill
}

\theoremstyle{definition}

\theoremstyle{plain}

\theoremstyle{remark}
\newtheorem{cor}{Corollary}
\theoremstyle{plain}

\pgfplotsset{
    /pgfplots/ybar legend/.style={
    /pgfplots/legend image code/.code={%
       \draw[##1,/tikz/.cd,yshift=-0.25em]
        (0cm,0cm) rectangle (6pt,0.7em);},
   },
}

\DeclareMathOperator{\dplus}{+\kern -0.45em+}

\newcommand\mydots{\hbox to 1em{.\hss.\hss.}}

\usepackage{fleqn}
\usepackage{math}
\usepackage{bcprules}
\usepackage{prooftree}

\newcommand{\biggap}{\quad\quad\quad}

\newcommand{\ba}{\begin{array}}
\newcommand{\ea}{\end{array}}

\newcommand{\ei}{\end{array}}
\newcommand{\bcases}{\left\{\begin{array}{ll}}
\newcommand{\ecases}{\end{array}\right.}

\newcommand{\defmbox}[1]{ \mbox{\biggap #1} }

\usepackage{wrapfig}

\newtheoremstyle{definition}
  {2.5pt plus 0.5pt minus 0.5pt} 
  {2.5pt plus 0.5pt minus 0.5pt} 
  {} 
  {} 
  {\bfseries} 
  {.} 
  {.5em} 
  {} 

\theoremstyle{definition}

\usepackage{comment}
\usepackage{lipsum}

\newcommand{\paratitle}[1]{\textbf{#1}.}
\newcommand{\m}[1]{\texttt{#1}}
\newcommand{\proofcase}[1]{\vspace{1mm}\textit{Case} \textsc{#1}:}

\usepackage{newunicodechar}
\newunicodechar{€}{\monospaced}
\def\monospaced#1€{$\texttt{#1}$}

\usepackage{pifont}
\title{On the Soundness of Coroutines with Snapshots}

\author[1]{Aleksandar Prokopec}
\author[2]{Fengyun Liu}
\affil[1]{Oracle Labs, Switzerland\\
  \texttt{aleksandar.prokopec@oracle.com}}
\affil[2]{École Polytechnique Fédérale de Lausanne, Switzerland\\
  \texttt{fengyun.liu@epfl.ch}}
\authorrunning{A.\, Prokopec and F.\, Liu} 

\Copyright{Aleksandar Prokopec and Fengyun Liu}

\keywords{coroutines, continuations, asynchronous programming, inversion of control}

\EventEditors{John Q. Open and Joan R. Acces}
\EventNoEds{2}
\EventLongTitle{Europe (ECOOP 2018)}
\EventShortTitle{ECOOP 2018}
\EventAcronym{ECOOP}
\EventYear{2018}
\EventDate{July 16-21, 2018}
\EventLocation{Amsterdam, Netherlands}
\EventLogo{}
\SeriesVolume{42}
\ArticleNo{23}

\begin{document}

\maketitle

\begin{abstract}
\emph{Coroutines} are a general control flow construct
that can eliminate control flow fragmentation
inherent in event-driven programs,
which is still missing in many popular languages.
Coroutines with snapshots are
a first-class, type-safe, stackful coroutine model,
which unifies many variants of suspendable computing,
and is sufficiently general to express
iterators, single-assignment variables, async-await,
actors, event streams,
backtracking, symmetric coroutines and continuations.

In this paper, we develop a formal model called $\lambda_\rightsquigarrow$
that captures the essence of type-safe, stackful, delimited coroutines with snapshots.
We prove the standard progress and preservation safety properties.
Finally, we show a formal transformation
from the $\lambda_\rightsquigarrow$ calculus
to the simply-typed lambda calculus with references.
\end{abstract}





\section{Introduction}

Asynchronous programming is becoming increasingly important,
with applications ranging from
actor systems \cite{akka,Haller:2009:SAU:1496391.1496422},
futures and network programming
\cite{Cantelon:2013:NA:2601501,SIP14},
user interfaces \cite{EPFL-REPORT-176887},
to functional stream processing
\cite{Meijer:2012:YMD:2160718.2160735}.
Traditionally, these programming models
were realized either by blocking execution threads
(which can be detrimental to performance \cite{1327913}),
or callback-style APIs
\cite{Cantelon:2013:NA:2601501,SIP14,johnson1988designing},
or with monads
\cite{Wadler:1995:MFP:647698.734146}.
However, these approaches often feel unnatural,
and the resulting programs can be hard to understand and maintain.
\textit{Coroutines} \cite{Conway:1963:DST:366663.366704}
overcome the need for blocking threads, callbacks and monads
by allowing parts of the execution to pause at arbitrary points,
and resuming that execution later.

In related work \cite{prokopec2018coroutines},
we introduced \emph{stackful coroutines with snapshots},
and showed that they are sufficiently general to express
iterators, single-assignment variables,
async-await, actors, event streams, backtracking,
symmetric coroutines and continuations.
Additionally, we provided an efficient implementation of coroutines with snapshots
for Scala \cite{coroutines}.

In this paper, we develop a formal model called $\lambda_\rightsquigarrow$
that captures the essence of type-safe, stackful, delimited coroutines with snapshots.
We prove the standard progress and preservation safety properties.
Finally, we show a formal transformation
from the $\lambda_\rightsquigarrow$ calculus
to the simply-typed lambda calculus with references.

\section{Simply Typed Lambda Calculus with Coroutines ($\lambda_\leadsto$)}
\label{sec:simply-typed-coroutines}

In addition to the standard abstraction, application and variable terms
associated with the lambda calculus,
the $\lambda_\leadsto$ extension defines coroutines and the accompanying operations.
A coroutine can be defined, called, resumed, suspended and copied.
A coroutine definition is similar to a function definition,
the main difference being that the body of the coroutine can suspend itself.
Calling a coroutine creates an invocation value of that coroutine,
which is initially paused.
When resumed, that invocation runs until suspending, or completing.
A coroutine instance is suspended whenever it evaluates the \texttt{yield} term
(when this happens, we say that the evaluation \emph{yields}),
and can be subsequently resumed from the same point.
The state of a coroutine instance can also be duplicated into a new instance,
which is referred to as creating a \emph{snapshot}.

A coroutine instance is represented not just by the respective term,
but also by its suspended computation state
(a coroutine \emph{instance} is a \emph{stateful} entity).
Multiple variables can refer to (i.e. alias) the same coroutine instance,
which holds some mutable state.
For this reason, coroutine instances are represented
with special instance labels $i$.
Each instance label has a corresponding mapping in the \emph{coroutine store},
as we will show.

Importantly, in the model that we will define,
coroutines are \emph{delimited}.
This means that yielding (i.e. suspending a coroutine)
must take place inside the program region that is specially marked as a coroutine,
and only such regions of the program can get suspended.
This region is \emph{not} necessarily lexically scoped,
since a coroutine definition can call into another coroutine definition
(defined elsewhere).
In other words, we model stackful delimited coroutines.

We start by defining the syntax for $\lambda_\leadsto$,
which consists of the user terms, and the runtime terms
(i.e. terms that arise only during evaluation).

\begin{definition}[Syntax]\label{def:syntax}
The $\lambda_\leadsto$ \emph{programming model}
consists of the following terms,
which can appear in user programs:

\begin{minipage}[b]{12 cm}
$\ba[t]{l@{\hspace{15mm}}r}
\texttt{t} ::= & \defmbox{user terms:} \\
\biggap \texttt{(x:T) => t}
& \defmbox{abstraction}  \\
\biggap \texttt{t(t)}
& \defmbox{application}  \\
\biggap \texttt{x}
& \defmbox{variable}  \\
\biggap \texttt{()}
& \defmbox{unit value}  \\
\biggap \texttt{(x:T)} \overset{\texttt{T}}{\rightsquigarrow} \texttt{t}
& \defmbox{coroutine}  \\
\biggap \texttt{yield(t)}
& \defmbox{yielding}  \\
\biggap \texttt{start(t, t)}
& \defmbox{coroutine instance creation}  \\
\biggap \texttt{resume(t, v, v, v)}
& \defmbox{resuming}  \\
\biggap \texttt{snapshot(t)}
& \defmbox{snapshot creation}  \\
\biggap \texttt{fix(t)}
& \defmbox{recursion}  \\
\ea$
\end{minipage}
\\*

The $\lambda_\leadsto$ model defines the following types:

\begin{minipage}[b]{12 cm}
$\ba[t]{l@{\hspace{42mm}}r}
\texttt{T} ::= & \defmbox{types:} \\
\biggap \texttt{T => T}
& \defmbox{function type}  \\
\biggap \texttt{T} \overset{\texttt{T}}{\rightsquigarrow} \texttt{T}
& \defmbox{coroutine type}  \\
\biggap \texttt{T} \leftrightsquigarrow \texttt{T}
& \defmbox{coroutine instance type}  \\
\biggap \texttt{Unit}
& \defmbox{unit type}  \\
\biggap \bot
& \defmbox{bottom type}  \\
\ea$
\end{minipage}
\\*

The $\lambda_\leadsto$ \emph{model of computation}
additionally defines the following \emph{runtime terms},
which cannot appear in a user program,
but can appear during the evaluation of a program:

\begin{minipage}[b]{12 cm}
$\ba[t]{l@{\hspace{36mm}}r}
\texttt{r} ::= & \defmbox{runtime terms:} \\
\biggap i
& \defmbox{coroutine instance}  \\
\biggap \langle \texttt{t, v, v, v} \rangle_i
& \defmbox{coroutine resumption}  \\
\biggap \llbracket \texttt{t} \rrbracket_\texttt{v}
& \defmbox{suspension}  \\
\biggap \varnothing
& \defmbox{empty term}  \\
\ea$
\end{minipage}
\\*

The following subset of terms are considered \emph{values}:

\begin{minipage}[b]{12 cm}
$\ba[t]{l@{\hspace{42mm}}r}
\texttt{v} ::= & \defmbox{values:} \\
\biggap \texttt{(x:T) => t}
& \defmbox{abstraction}  \\
\biggap \texttt{()}
& \defmbox{unit value}  \\
\biggap \texttt{(x:T)} \overset{\texttt{T}}{\rightsquigarrow} \texttt{t}
& \defmbox{coroutine}  \\
\biggap i
& \defmbox{coroutine instance}  \\
\biggap \varnothing
& \defmbox{empty term}  \\
\ea$
\end{minipage}
\\*
\qed
\end{definition}

We once more highlight the difference between a \emph{coroutine},
which is akin to a function definition,
and a \emph{coroutine instance} which is aking to an invocation of a function.
Since a coroutine instance, unlike a function invocation,
can be suspended and resumed,
it must be a first-class value that the program can refer to.
We therefore distinguish between a coroutine type
$\texttt{T}_1 \overset{\texttt{T}_y}{\rightsquigarrow} \texttt{T}_2$
(where $\texttt{T}_1$ is the input type,
$\texttt{T}_2$ is the return type,
and $\texttt{T}_y$ is the yield type)
and the coroutine instance type $\texttt{T}_y \leftrightsquigarrow \texttt{T}_2$
(where $\texttt{T}_y$ is the yield type,
and $\texttt{T}_2$ is the return type).

Before defining the typing rules for $\lambda_\leadsto$,
we first define the contexts in which the typing of a term takes place.
There are two kinds of contexts that we need --
the first is the standard typing context $\Gamma$ used to track variable types,
and the second is a \emph{instance typing} $\Sigma$,
used to track the types of the coroutine instances that exist at runtime.

\begin{definition}[Typing context]\label{def:typing-context}
The \emph{typing context} $\Gamma$
is a sequence of variables and their respective types,
where the comma operator ($,$) extends a typing context with a new binding:

\begin{minipage}[b]{12 cm}
$\ba[t]{l@{\hspace{54mm}}r}
\Gamma ::= & \defmbox{typing context:} \\
\biggap \varnothing
& \defmbox{empty context}  \\
\biggap \Gamma , \texttt{x:T}
& \defmbox{variable binding}  \\
\ea$
\end{minipage}
\\*
\qed
\end{definition}

\begin{definition}[Instance typing]\label{def:instance-typing}
The \emph{instance typing} $\Sigma$
is a sequence of coroutine instance labels and their respective types,
where the comma operator ($,$) extends an instance typing with a new binding:

\begin{minipage}[b]{12 cm}
$\ba[t]{l@{\hspace{44mm}}r}
\Sigma ::= & \defmbox{instance typing:} \\
\biggap \varnothing
& \defmbox{empty instance typing}  \\
\biggap \Sigma , i\texttt{:T}
& \defmbox{new instance}  \\
\ea$
\end{minipage}
\\*
\qed
\end{definition}

Aside from tracking the type of each term,
our typing rules will track the type of values that a term can yield.
This allows typechecking coroutine declarations against
the values yielded in their bodies.
Therefore, our typing relation will be a five place relation
between the typing context, instance typing, the term and its type,
and the yield type.

\begin{definition}[Typing relation]\label{def:typing-rules}
The \emph{typing relation}
$\Sigma | \Gamma \vdash \texttt{t:T} | \texttt{T}_y$
on $\lambda_\leadsto$
is a relation between the instance typing $\Sigma$,
the typing context $\Gamma$,
the term $t$,
the type of the term $T$,
and the \emph{yield type} $T_y$,
where $T_y$ denotes the type of values that may be yielded during the
evaluation of the term $t$.
The inductive definition of this typing relation
is shown in Fig. \ref{fig:typing-relation} and Fig. \ref{fig:runtime-typing-relation}.
\qed
\end{definition}

\begin{figure}[t]
\input{typing-rules.tex}
\caption{Typing relation on terms}
\label{fig:typing-relation}
\end{figure}

We now briefly discuss the typing rules in Fig. \ref{fig:typing-relation}.
The rules \textsc{T-Abs}, \textsc{T-App}, \textsc{T-Var} and \textsc{T-Unit}
are standard in simply typed lambda calculus.
In our case, we base the typing judgement
on the instance typing $\Sigma$ in addition to the typing context $\Gamma$.
Furthermore, each typing judgement has the yield type as the last element.
The rule \textsc{T-App} allows the evaluation of the lambda \texttt{t}$_1$
and the argument \texttt{t}$_2$ to yield values of some type \texttt{T}$_y$
(but the type \texttt{T}$_y$ must be the same for both terms).
The rule \textsc{T-Var} assigns the bottom type $\bot$
as the yield type of a variable,
since this term does not yield any values.

Note that the rule \textsc{T-Abs} requires that the body of the lambda
has the yield type $\bot$, that is, does not yield any values.
The reason for this is that we model \emph{delimited} coroutines.
Yielding is only allowed from within a lexical scope of a coroutine definition,
that is, a yielded value cannot cross a function boundary.
In our concrete implementation,
this means that we allow defining a coroutine that normally invokes
a 3rd party library function,
but the body of that 3rd party library function is not allowed to yield
(unless the 3rd party library function is itself a coroutine
defined using our transformation).

The rule \textsc{T-Ctx} says that any term \texttt{t}
of type \texttt{T} and the yield type $\bot$
can be assumed to have any yield type \texttt{T}$_y$.
A typechecker may apply this rule when consolidating
two terms one of which does not yield.
For example,
given a term \texttt{t}$_1$ whose yield type is \texttt{T}$_y \neq \bot$,
and a term \texttt{t}$_2$ whose yield type is $\bot$,
a typechecker must apply \textsc{T-Ctx} before applying the rule \textsc{T-App}.


The rule \textsc{T-Coroutine}
types a coroutine declaration.
Similar to how \textsc{T-Abs} types a lambda, this rule types
a coroutine using the return type and the yield type of its body.
The yield type is ``swallowed'' by the coroutine,
leaving $\bot$ as the yield type of the resulting value.

Having inspected the rules for typing a coroutine declaration,
we turn to coroutine operations.
The rule \textsc{T-Start} says that
given a term of the coroutine type
$\texttt{T}_1 \overset{\texttt{T}_y}{\rightsquigarrow} \texttt{T}_2$,
and a term of type $\texttt{T}_1$,
a $\texttt{start}$ expression has the coroutine instance type
$\texttt{T}_y \leftrightsquigarrow \texttt{T}_2$.
The evaluation of the \texttt{start} term
can itself yield a value of type \texttt{T}$_w$,
independently of the yield type of the newly created coroutine instance --
this means that the yielding context of the coroutine instance
is separated from the context of the caller,
where the caller is either the body of the enclosing coroutine,
or the enclosing function, or the top-level program.

The \textsc{T-Yield} rule says that
if a term \texttt{t} has a type \texttt{T},
and its yield type is also \texttt{T},
then \texttt{yield(t)} has the type \texttt{Unit},
with the same yield type \texttt{T}.
In other words,
a coroutine can yield a value of type \texttt{T},
only if all the previous yields were of the same type \texttt{T}
(or, if there were not previous yields,
the typechecker can apply the \textsc{T-Ctx} rule).

The \textsc{T-Resume} rule describes the type of the \texttt{resume} expression.
Consider resuming a coroutine instance of type
$\texttt{T}_y \leftrightsquigarrow \texttt{T}_2$.
Depending on the state of the instance,
this has several outcomes.
First, the instance can complete and return a result value of type \texttt{T}$_2$.
Second, the resumption of the instance can suspend itself and yield a value
of type \texttt{T}$_y$.
Finally, the instance could have already been completed when \texttt{resume}
was called.
Since $\lambda_\leadsto$ does not have a variants or sum types
to distinguish between these cases,
the \texttt{resume} statement acts as a proverbial poor man's pattern matching.
The values \texttt{t}$_2$, \texttt{t}$_3$ and \texttt{t}$_4$
represent the code segments that deal with each of the above-described cases,
and they return a result value of type \texttt{T}$_R$.
Note that we could have modeled \texttt{t}$_2$, \texttt{t}$_3$ and \texttt{t}$_4$
as function values.
However, that would mean that the bodies of
\texttt{t}$_2$, \texttt{t}$_3$ and \texttt{t}$_4$
cannot yield values themselves (according to \textsc{T-Abs}).
Hence, we model \texttt{t}$_2$, \texttt{t}$_3$ and \texttt{t}$_4$
as coroutines with the yield type \texttt{T}$_w$,
which corresponds to the yield type of the enclosing context.

The \textsc{T-Snapshot} rule types the \texttt{snapshot} expression,
which copies the given coroutine instance.
The coroutine instance type and the yield type are preserved
between the premise and the conclusion.

We mentioned that this model must describe stackful coroutines.
To this end, a coroutine body must be able to invoke another coroutine
as if it were a normal function.
If that other function then yields,
both the caller and the callee must be suspended.
Note that coroutine invocation is syntactically equivalent to function application,
but the callee is a coroutine, not a function.
This is captured by the \textsc{T-AppCor} rule,
which additionally requires that the yield type \texttt{T}$_y$
of the callee corresponds to the yield type of the current coroutine.

Finally, \textsc{T-Fix} is a standard typing rule for general recursion,
which allows a coroutine or a lambda to refer to itself.
We defer the discussion of the typing rules for runtime terms,
shown in Fig. \ref{fig:runtime-typing-relation},
until we cover the operational semantics of $\lambda_\leadsto$.
Before we proceed, we define what it means for a program to be well-typed.

\begin{definition}[Well-typed program]\label{def:well-typed-term}
A term \texttt{t} is \emph{well-typed} if and only if
$\exists \texttt{T}, \texttt{T}_y$ and an instance typing $\Sigma$ such that
$\Sigma | \varnothing \vdash \texttt{t:T} | \texttt{T}_y$.
Furthermore, a term \texttt{t} is a \emph{well-typed user program}
if \texttt{t} is well-typed and its yield type $\texttt{T}_y = \bot$.
\qed
\end{definition}

Coroutine instances are stateful --
each coroutine instance maps to a term that it evaluates.
Program evaluation is modeled not only as a transition between terms,
but also between coroutine stores $\mu$, which we define next.

\begin{definition}[Coroutine store]\label{def:coroutine-store}
A \emph{coroutine store} $\mu$ is a sequence of
coroutine instance labels $i$ bound to respective evaluation terms \texttt{t},
where the comma operator ($,$) extends the coroutine store
with a new binding.

\begin{minipage}[b]{12 cm}
$\ba[t]{l@{\hspace{44mm}}r}
\mu ::= & \defmbox{instance store:} \\
\biggap \varnothing
& \defmbox{empty instance store}  \\
\biggap \mu, i \triangleright \texttt{t}
& \defmbox{instance binding}  \\
\ea$
\end{minipage}
\\*
\qed
\end{definition}

In what follows, we will use the convention
that the evaluation term \texttt{t} in the coroutine store is suspended
($\texttt{t} = \llbracket \texttt{t}_1 \rrbracket_\varnothing$)
if and only if the coroutine
is currently executing, or has completed altogether.
Otherwise, if the coroutine can be resumed,
the evaluation term \texttt{t} will not be suspended.

\begin{definition}[Well-typed coroutine store]\label{lem:well-typed-store}
A coroutine store $\mu$ is well-typed with respect to the instance typing $\Sigma$,
denoted $\Sigma \vdash \mu$,
if and only if it is true that $\forall i \in dom(\mu)$,
$\Sigma(i) = \texttt{T}_y \leftrightsquigarrow \texttt{T}_2
\Leftrightarrow
\Sigma | \varnothing \vdash \mu(i) \texttt{:T}_2 | \texttt{T}_y$,
and $dom(\Sigma) = dom(\mu)$.
\qed
\end{definition}

We next define the operational semantics of $\lambda_\leadsto$.

\begin{definition}[Transition relation]\label{def:transition-rules}
The \emph{transition relation} $\texttt{t}|\mu \rightarrow \texttt{t}'|\mu'$
is a four place relation between
the source term $\texttt{t}$ and source coroutine store $\mu$
and the target term $\texttt{t}'$ and the target coroutine store $\mu'$.
The inductive definition of the transition relation
is shown in Fig. \ref{fig:evaluation-rules}.
\qed
\end{definition}

\begin{figure}

\begin{tabular}{lll}
$E$ & ::= & [$\cdot$] | $E€(t)€$ | $€v(€E€)€$ | $€start(€E€,t)€$ | $€start(v,€E€)€$ | $€yield(€E€)€$ | $€snapshot(€E€)€$ | $€fix(€E€)€$ | \\
  &     & $€resume(€E€,t,t,t)€$ | $€resume(v,€E€,t,t)€$ | $€resume(v,v,€E€,t)€$ | $€resume(v,v,v,€E€)€$ | \\
  &     & $\langle E, €v€, €v€, €v€ \rangle_i$
\end{tabular}

\begin{tabular}{lll}
$P$ & ::= & $[\cdot]€(t)€$ | $€v(€[\cdot]€)€$ | $€start(€[\cdot]€,t)€$ | $€start(v,€[\cdot]€)€$ | $€yield(€[\cdot]€)€$ | $€snapshot(€[\cdot]€)€$ | $€fix(€[\cdot]€)€$ | \\
  &     & $€resume(€[\cdot]€,t,t,t)€$ | $€resume(v,€[\cdot]€,t,t)€$ | $€resume(v,v,€[\cdot]€,t)€$ | $€resume(v,v,v,€[\cdot]€)€$
\end{tabular}

\infrule[E-Context]
{
€t€ | \mu
\rightarrow
€t€' | \mu'
}
{
E[€t€] | \mu
\rightarrow
E[€t€'] | \mu'
}

\infax[E-Pause]{
P[\llbracket €t€ \rrbracket_\m{v}] | \mu
\rightarrow
\llbracket P[€t€] \rrbracket_\m{v} | \mu
}

\infax[E-AppAbs]{
  €((x:T)=>t)(v)€ | \mu \rightarrow [€x€ \mapsto €v€] €t€ | \mu
}

\infrule[E-Start]
{
i \not\in dom(\mu)
}
{
€start((x:T€_1€)€\overset{\m{T}_y}{\rightsquigarrow}€t,v)€ | \mu
\rightarrow
i | \mu , i \triangleright [€x€ \mapsto €v€]€t€
}

\infax[E-Yield]{
€yield(v)€ | \mu \rightarrow \llbracket €()€ \rrbracket_\m{v} | \mu
}

\infrule[\textsc{E-Snapshot}]
{
i_2 \not\in dom(\mu)
\quad
i_1 \neq i_2
}
{
€snapshot(€i_1€)€ | \mu , i_1 \triangleright €t€
\rightarrow
i_2 | \mu , i_1 \triangleright €t€ , i_2 \triangleright €t€
}

\infax[E-AppCor]{
€((x:T€_1 \overset{\m{T}_y}{\rightsquigarrow} €t€_2€)(v)€ | \mu
\rightarrow
[€x€ \mapsto €v€] €t€_2 | \mu
}

\infax[E-Fix]{
€fix((x:T€_1€)=>t€_2€)€ | \mu
\rightarrow
[€x€ \mapsto €fix((x:T€_1€)=>t€_2€)€] €t€_2 | \mu
}

\infrule[E-Resume1]
{
€t€ \neq \llbracket €t€_0 \rrbracket_\varnothing
}
{
€resume(€i€,v€_2€,v€_3€,v€_4€)€
| \mu , i \triangleright €t€
\rightarrow
\langle t, v_2, v_3, v_4 \rangle_i
| \mu , i \triangleright \llbracket t \rrbracket_\varnothing
}

\infax[E-Resume2]{
€resume(€i€,v€_2€,v€_3€,v€_4€)€
| \mu , i \triangleright \llbracket €t€_0 \rrbracket_\varnothing
\rightarrow
€v€_4€(())€
| \mu , i \triangleright \llbracket €t€_0 \rrbracket_\varnothing
}

\infax[E-Capture]{
\langle \llbracket €t€_1 \rrbracket_\m{v}€,v€_2€,v€_3€,v€_4 \rangle_i
| \mu , i \triangleright \llbracket €t€_0 \rrbracket_{\m{v}'}
\rightarrow
€v€_3€(v)€ | \mu , i \triangleright €t€_1
}

\infax[E-Terminate]{
\langle €v,v€_2€,v€_3€,v€_4 \rangle_i
| \mu , i \triangleright \llbracket €t€_0 \rrbracket_{\m{v}'}
\rightarrow
€v€_2€(v)€
| \mu , i \triangleright \llbracket €v€ \rrbracket_\varnothing
}

\caption{Transition relation}
\label{fig:evaluation-rules}
\end{figure}

For simplicity, the evaluation rules are presented with the evaluation
context $E$ and suspension context $P$. The evaluation
context $E$ is standard, and it's used in \textsc{E-Context}. The suspension
context $P$ is used in \textsc{E-Pause} to simplify evaluations like the following:

\infax{
\llbracket \texttt{t}_1 \rrbracket_\texttt{v} \texttt{(t}_2\texttt{)} | \mu
\rightarrow
\llbracket \texttt{t}_1\texttt{(t}_2\texttt{)} \rrbracket_\texttt{v} | \mu \quad
}

\infax{
\texttt{t}_1 \texttt{(} \llbracket \texttt{t}_2 \rrbracket_\texttt{v} \texttt{)} | \mu
\rightarrow
\llbracket \texttt{t}_1\texttt{(t}_2\texttt{)} \rrbracket_\texttt{v} | \mu
}

\infax{
\texttt{yield(} \llbracket \texttt{t} \rrbracket_\texttt{v} \texttt{)} | \mu
\rightarrow
\llbracket \texttt{yield(t)} \rrbracket_\texttt{v} | \mu
}

We briefly discuss the evaluation rules in Fig. \ref{fig:evaluation-rules}.
The rule \textsc{E-AppAbs} is standard in lambda calculus.
The only difference in our case is the addition
of the coroutine store, which does not change in \textsc{E-AppAbs}.

The rule \textsc{E-Start} reduces a \texttt{start} expression
with a coroutine and its argument to a coroutine instance
with a fresh label $i$,
and adds a binding from $i$ to the coroutine body \texttt{t}
in which occurrences of \texttt{x} are replaced with the argument \texttt{v}.

The rule \textsc{E-Yield} reduces the \texttt{yield} expression
to a \emph{suspension} of the \texttt{Unit} value \texttt{()},
which has a pending yield of the value \texttt{v}.
A term suspension is a runtime term of the form
$\llbracket \texttt{t} \rrbracket_\texttt{v}$,
where \texttt{t} represents the suspended computation,
and \texttt{v} is the value that is about to be yielded.
Once reduced from a \texttt{yield} expression,
a suspension spreads through the program until reaching the limits
of the enclosing coroutine resumption.
To model this, we need to introduce evaluation rules
that suspend all term shapes.
The rule \textsc{E-Pause} exists for this purpose.


The expanding suspension is captured once it reaches the
\emph{coroutine resumption term}
$\langle \llbracket \texttt{t} \rrbracket_\texttt{v},
\texttt{v}_2,
\texttt{v}_3,
\texttt{v}_4 \rangle_i$,
as described by the rule \textsc{E-Capture}.
The suspended execution term \texttt{t}
from the suspension $\llbracket \texttt{t} \rrbracket_\texttt{v}$
is placed into the coroutine store binding of the coroutine instance $i$,
and the yielded value \texttt{v} is passed
to the yield-handling function \texttt{v}$_3$.

Consider the outcomes of resuming a coroutine.
If a coroutine instance $i$ is not terminated and not currently executing
(that is, the binding for $i$ in the coroutine store does not point
to a suspended term $\llbracket \texttt{t}_0 \rrbracket_\varnothing$),
then a \texttt{resume} expression reduces
to a coroutine resumption term, by the rule \textsc{E-Resume1}.
It is illegal to reduce a coroutine instance that is already completed
or currently executing.
A \texttt{resume} expression on such a coroutine instance
reduces to an application of the fourth argument to a \texttt{Unit} term,
by the rule \textsc{E-Resume2}.

The rule \textsc{E-Terminate} states that
if the term \texttt{t} reduces to a value \texttt{v},
the coroutine resumption reduces to an application
of the second argument \texttt{v}$_2$ to the reduced value \texttt{v},
leaving the binding for the instance $i$
mapped to a suspended state $\llbracket \texttt{v} \rrbracket_\varnothing$.

Given an existing coroutine instance $i_1$,
the \textsc{E-Snapshot} rule reduces
a \texttt{snapshot(}$i_1$\texttt{)} expression
to a fresh coroutine instance $i_2$,
and adds a copy of the $i_1$'s term \texttt{t} to the store as a binding for $i_2$.
Finally,
note that $\lambda_\leadsto$ models stackful delimited coroutines,
so a coroutine application must be allowed inside the body of a coroutine.
This is shown in the rule \textsc{E-AppCor},
which essentially describes beta reduction on coroutines.

\begin{figure}
\input{typing-rules-runtime.tex}
\caption{Typing relation on runtime terms}
\label{fig:runtime-typing-relation}
\end{figure}

Now that we saw how normal terms reduce to runtime terms,
we can inspect the remaining typing derivations,
shown in Fig. \ref{fig:runtime-typing-relation}.
A suspension must have the same type \texttt{T} as the term \texttt{t} is suspends,
and the yield type \texttt{T}$_y$ that corresponds to the yielded value,
as stated by the rule \textsc{T-Suspension}.
The rule \textsc{T-Instance}
states that the type of the coroutine instance $i$ is
$\texttt{T}_y \leftrightsquigarrow \texttt{T}_2$,
under the assumption that the instance typing $\Sigma$
contains the corresponding binding for $i$.
\textsc{T-Empty} states that one can assume that the empty term has any type.
Finally, the rule \textsc{T-Resumption}
assigns a type to a coroutine resumption term,
and is a direct equivalent of the rule \textsc{T-Resume}
from Fig. \ref{fig:typing-relation}.

With the typing rules and the operational semantics in place,
we can prove the basic safety properties of the $\lambda_\leadsto$ model --
progress and preservation.
We start by establishing several helper lemmas.

\begin{lemma}[Inversion of the typing relation]\label{lem:inversion}
\begin{enumerate}
\item
If $\Sigma | \Gamma \vdash \texttt{():T} | \texttt{T}_y$,
then $\texttt{T} = \texttt{Unit}$.
\item
If $\Sigma | \Gamma \vdash \texttt{x:T} | \texttt{T}_y$,
then $\texttt{x:T} \in \Gamma$.
\item
If $\Sigma | \Gamma \vdash \texttt{(x:T}_1\texttt{)=>t}_2\texttt{:T} | \texttt{T}_y$,
then $\exists \texttt{T}_2, \texttt{T} = \texttt{T}_1 \texttt{=>} \texttt{T}_2$,
$\Sigma | \Gamma , \texttt{x:T}_1 \vdash \texttt{t}_2\texttt{:T}_2 | \bot$.
\item
If $\Sigma | \Gamma \vdash \texttt{t}_1\texttt{(t}_2\texttt{):T} | \texttt{T}_y$,
then either
$\exists \texttt{T}_2$
such that
$\Sigma | \Gamma \vdash \texttt{t}_1\texttt{:T}_1\texttt{=>T} | \texttt{T}_y$
and $\Sigma | \Gamma \vdash \texttt{t}_2\texttt{:T}_2 | \texttt{T}_y$,
or $\exists \texttt{T}_2$ such that
$\Sigma | \Gamma \vdash \texttt{t}_1
\texttt{:T}_1 \overset{\texttt{T}_y}{\rightsquigarrow} \texttt{T} | \texttt{T}_y$
and $\Sigma | \Gamma \vdash \texttt{t}_2\texttt{:T}_2 | \texttt{T}_y$.
\item
If $\Sigma | \Gamma \vdash
\texttt{(x:T}_1\texttt{)}\overset{\texttt{T}_y}{\rightsquigarrow}\texttt{t}_2
\texttt{:T} | \texttt{T}_w$,
then $\exists \texttt{T}_2,
\texttt{T} = \texttt{T}_1\overset{\texttt{T}_y}{\rightsquigarrow}\texttt{T}_2$,
$\Sigma | \Gamma , \texttt{x:T}_1 \vdash \texttt{t:T}_2 | \texttt{T}_y$.
\item
If $\Sigma | \Gamma \vdash \texttt{start(t}_1\texttt{,t}_2\texttt{):T} | \texttt{T}_w$,
then $\exists \texttt{T}_1, \texttt{T}_2, \texttt{T}_y$ such that
$\texttt{T} = \texttt{T}_y \leftrightsquigarrow \texttt{T}_2$,
$\Sigma | \Gamma \vdash \texttt{t}_1
\texttt{:T}_1 \overset{\texttt{T}_y}{\rightsquigarrow} \texttt{T}_2 | \texttt{T}_w$,
and $\Sigma | \Gamma \vdash \texttt{t}_2 \texttt{:T}_2 | \texttt{T}_w$.
\item
If $\Sigma | \Gamma \vdash \texttt{yield(t):T} | \texttt{T}_y$,
then $\texttt{T} = \texttt{Unit}$
and $\Sigma | \Gamma \vdash \texttt{t}\texttt{:T}_y | \texttt{T}_y$.
\item
If $\Sigma | \Gamma \vdash \texttt{snapshot(t):T} | \texttt{T}_w$,
then $\exists \texttt{T}_2, \texttt{T}_y$ such that
$\texttt{T} = \texttt{T}_y \leftrightsquigarrow \texttt{T}_2$,
and $\Sigma | \Gamma \vdash \texttt{t}
\texttt{:T}_y \leftrightsquigarrow \texttt{T}_2 | \texttt{T}_y$.
\item
If $\Sigma | \Gamma \vdash
\texttt{resume(t}_1\texttt{,t}_2\texttt{,t}_3\texttt{,t}_4\texttt{):T} | \texttt{T}_w$,
then $\exists \texttt{T}_2, \texttt{T}_y$ such that the following holds:
$\Sigma | \Gamma \vdash \texttt{t}_1
\texttt{:T}_y \leftrightsquigarrow \texttt{T}_2 | \texttt{T}_w$,
and
$\Sigma | \Gamma \vdash \texttt{t}_2 \texttt{:T}_2
\overset{\texttt{T}_w}{\rightsquigarrow} \texttt{T}_R
| \texttt{T}_w$,
and
$\Sigma | \Gamma \vdash \texttt{t}_3 \texttt{:T}_y
\overset{\texttt{T}_w}{\rightsquigarrow} \texttt{T}_R
| \texttt{T}_w$,
and $\Sigma | \Gamma \vdash \texttt{t}_4 \texttt{:Unit}
\overset{\texttt{T}_w}{\rightsquigarrow} \texttt{T}_R
| \texttt{T}_w$.
\item
If $\Sigma | \Gamma \vdash \texttt{fix(t):T} | \texttt{T}_w$,
then $\Sigma | \Gamma \vdash \texttt{t:T=>T} | \bot$.
\item
If $\Sigma | \Gamma \vdash \llbracket \texttt{t} \rrbracket_\texttt{v}
\texttt{:T} | \texttt{T}_y$,
then $\Sigma | \Gamma \vdash \texttt{t:T} | \texttt{T}_y$,
and $\Sigma | \Gamma \vdash \texttt{v:T}_y | \bot$.
\item
If $\Sigma | \Gamma \vdash i \texttt{:T} | \texttt{T}_w$,
then $\exists \texttt{T}_y, \texttt{T}_2$ such that
$i \texttt{:T}_y \leftrightsquigarrow \texttt{T}_2 \in \Sigma$.
\item
If $\Sigma | \Gamma \vdash
\langle \texttt{t}_1\texttt{,v}_2\texttt{,v}_3\texttt{,v}_4 \rangle_i
\texttt{:T}_R | \texttt{T}_w$,
then $\exists \texttt{T}_2, \texttt{T}_y$ such that the following holds:
$\Sigma | \Gamma \vdash \texttt{t}_1
\texttt{:T}_y \leftrightsquigarrow \texttt{T}_2 | \texttt{T}_w$,
$\Sigma | \Gamma \vdash \texttt{v}_2 \texttt{:T}_2
\overset{\texttt{T}_w}{\rightsquigarrow} \texttt{T}_R
| \texttt{T}_w$,
$\Sigma | \Gamma \vdash \texttt{v}_3 \texttt{:T}_y
\overset{\texttt{T}_w}{\rightsquigarrow} \texttt{T}_R
| \texttt{T}_w$,
and $\Sigma | \Gamma \vdash \texttt{v}_4 \texttt{:Unit}
\overset{\texttt{T}_w}{\rightsquigarrow} \texttt{T}_R
| \texttt{T}_w$.
\end{enumerate}
\end{lemma}

\begin{proof}
Follows immediately from the typing derivations in \ref{def:typing-rules}.
\end{proof}

\begin{lemma}[Canonical forms]\label{lem:canonical}
\begin{enumerate}
\item
If \texttt{v} is a value of type \texttt{Unit}, then \texttt{v} is \texttt{()}.
\item
If \texttt{v} is a value of type $\texttt{T}_1 \texttt{=>} \texttt{T}_2$,
then $\texttt{v} = \texttt{(x:T}_1\texttt{)=>t}_2$.
\item
If \texttt{v} is a value of type
$\texttt{T}_1 \overset{\texttt{T}_y}{\rightsquigarrow} \texttt{T}_2$,
then $\texttt{v} =
\texttt{(x:T}_1\texttt{)} \overset{\texttt{T}_y}{\rightsquigarrow} \texttt{t}_2$.
\item
If \texttt{v} is a value of type
$\texttt{T}_y \leftrightsquigarrow \texttt{T}_2$,
then $\texttt{v} = i$, where $i$ is an instance label.
\end{enumerate}
\end{lemma}

\begin{proof}
We consider the possible forms of values as per syntax from \ref{def:syntax},
and rely on \ref{lem:inversion} to prove the claim.

For example, the value \texttt{()} immediately satisfies the claim,
by \textsc{T-Unit}.
From \ref{lem:inversion}, we see that other types of values
(functions, coroutines and coroutine instances)
never have the type \texttt{Unit}.
The remaining cases are proved in a similar way.
\end{proof}

We can now state the progress property of $\lambda_\leadsto$.

\begin{theorem}[Progress]\label{thm:progress}
Suppose that \texttt{t} is a closed, well-typed term
for some \texttt{T} and $\Sigma$,
as defined in \ref{def:well-typed-term}.
Then, either \texttt{t} is a value,
or \texttt{t} is a suspension term
$\llbracket \texttt{t} \rrbracket_{\texttt{v}}$,
or, for any store store $\mu$ such that $\Sigma \vdash \mu$,
there is some term $\texttt{t}'$ and store $\mu'$ such that
$\texttt{t} | \mu \rightarrow \texttt{t}' | \mu'$.
\end{theorem}

\begin{proof}
Since \texttt{t} is well-typed, we proceed casewise on the typing derivations.
Cases \textsc{T-Abs}, \textsc{T-Unit}, \textsc{T-Coroutine}
and \textsc{T-Instance} follow directly,
since \texttt{t} is a value.
Cases \textsc{T-Var}, \textsc{T-Ctx}, \textsc{T-Suspension} and \textsc{T-Empty} are trivial.

For the sake of simplicity, in most of the following cases when we use the induction hypothesis,
we only consider the case where a subterm is a value,
ignoring the case where the subterm is a suspension ($\llbracket \texttt{t} \rrbracket_{\texttt{v}}$),
and the case where the subterm can take a step.
This is valid because of the rule \textsc{E-Pause} and \textsc{E-Context}, which can
be used to make $t$ take one step when the subterm is not a value.

\proofcase{T-App}
$€t€ = €t€_1€(t€_2€)€ $

We only consider the case where both $€t€_1$ and $€t€_2$ are values,
then \ref{lem:canonical} tells us
that $€t€_1 = €(x:T€_2€)=>t€_{11}$,
so \textsc{E-AppAbs} applies to \texttt{t}.

\proofcase{T-Start}
$€t€ = €start(t€_1€,t€_2€)€ $

Similar to \textsc{T-App}, but use \textsc{E-Start}.

\proofcase{T-Yield}
$€t€ = €yield(t€_1€)€ $

Similar to \textsc{T-App},
but we rely on on \textsc{E-Yield}.

\proofcase{T-Snapshot}
$€t€ = €snapshot(t€_1€)€ $

Similar to \textsc{T-Yield},
but use \textsc{E-Snapshot}.

\proofcase{T-Resume}
$€t€ = €resume(t€_1€,t€_2€,t€_3€,t€_4€)€ $

Only consider the case where all subterms are values.
By \ref{lem:canonical}, €t€$_1$ must be an instance label $i$.
In that case,
there exists a store $\mu$ that contains a binding for $i$,
such that either \textsc{E-Resume1} or \textsc{E-Resume2} applies.

\proofcase{T-AppCor}
$€t€ = €t€_1€(t€_2€)€ $

Similar to \textsc{T-App}, but use \textsc{E-AppCor}.

\proofcase{T-Fix}
$€t€ = €fix(t€_1€)€ $

Similar to \textsc{T-App}, but use \textsc{E-Fix}.

\proofcase{T-Resumption}
$€t€ = \langle €t€_1€,v€_2€,v€_3€,v€_4 \rangle_i$

If \texttt{t}$_1$ is a value, then \textsc{E-Terminate} applies.
If \texttt{t}$_1$ is a suspension, then \textsc{E-Capture} applies.
Otherwise, \texttt{t}$_1$ reduces by the induction hypothesis,
so \texttt{t} reduces by \textsc{E-Context}.

\end{proof}

Before we prove that types are preserved during evaluation,
we state several standard helper lemmas.

\begin{lemma}[Permutation]\label{lem:permutation}
Assume that $\Sigma | \Gamma \vdash \texttt{t:T} | \texttt{T}_y$,
and that $\Sigma'$ and $\Gamma'$ are permutations of
$\Sigma$ and $\Gamma$, respectively.
Then $\Sigma' | \Gamma' \vdash \texttt{t:T} | \texttt{T}_y$.
\end{lemma}

\begin{proof}
The proof is a straightforward induction on the typing rules.
\end{proof}

\begin{lemma}[Weakening]\label{lem:weakening}
Assume that $\Sigma | \Gamma \vdash \texttt{t:T} | \texttt{T}_y$.
Then it holds that for any $\Sigma' \supseteq \Sigma, \Gamma' \supseteq \Gamma$,
$\Sigma' | \Gamma' \vdash \texttt{t:T} | \texttt{T}_y$.
\end{lemma}

\begin{proof}
The proof is a straightforward induction on the typing rules.
\end{proof}

\begin{lemma}[Substitution]\label{lem:substitution}
If $\Sigma | \Gamma , \texttt{x:S} \vdash \texttt{t:T} | \texttt{T}_y$,
and $\Sigma | \Gamma \vdash \texttt{s:S} | \bot$,
then
$\Sigma | \Gamma \vdash [\texttt{x} \mapsto \texttt{s}] \texttt{t:T} | \texttt{T}_y$.
\end{lemma}

\begin{proof}
The proof is a straightforward induction on the typing rules.
The more interesting cases are
\textsc{T-Abs}, \textsc{T-Var} and \textsc{T-Coroutine},
and they rely on \ref{lem:permutation} and \ref{lem:weakening}.
We show the case for \textsc{T-Coroutine},
where
$\texttt{t} =
\texttt{(y:T}_2\texttt{)}\overset{\texttt{T}_y}{\rightsquigarrow}\texttt{t}_1$,
$\texttt{T} = \texttt{T}_2\overset{\texttt{T}_y}{\rightsquigarrow}\texttt{T}_1$,
and
$\Sigma | \Gamma , \texttt{x:S} , \texttt{y:T}_2 \vdash
\texttt{t}_1\texttt{:T}_1 | \texttt{T}_y$.
By applying permutation, we obtain
$\Sigma | \Gamma , \texttt{y:T}_2 , \texttt{x:S} \vdash
\texttt{t}_1\texttt{:T}_1 | \texttt{T}_y$.
By applying weakening,
we obtain $\Sigma | \Gamma , \texttt{y:T}_2 \vdash \texttt{s:S} | \texttt{T}_y$.
We use the last two results with the induction hypothesis
to obtain
$\Sigma | \Gamma , \texttt{y:T}_2 \vdash
[\texttt{x} \mapsto \texttt{s}] \texttt{t}_1\texttt{:T}_1 | \texttt{T}_y$.
Finally, from \textsc{T-Coroutine}, we get
$\Sigma | \Gamma , \texttt{y:T}_2 \vdash
\texttt{(y:T}_2\texttt{)}\overset{\texttt{T}_y}{\rightsquigarrow}
[\texttt{x} \mapsto \texttt{s}] \texttt{t}_1\texttt{:T}_1 | \bot$.
Since $\texttt{x} \neq \texttt{y}$ (we can rename coroutine variables as needed),
the result follows.
\end{proof}

\begin{lemma}[Suspension]\label{lem:suspension}
If $\Sigma | \Gamma \vdash P[\llbracket \texttt{t} \rrbracket_\texttt{v}] : €T€ | \texttt{T}_y$,
then
If $\Sigma | \Gamma \vdash \llbracket P[\texttt{t}] \rrbracket_\texttt{v} : €T€ | \texttt{T}_y$.
\end{lemma}

\begin{proof}
By cases analysis on the suspension context $P$.
We only show the case for application, other cases are similar.

\proofcase{$P = [\cdot]€(t€_1€)€ $}

From the typing rule \textsc{T-App}, we have
$\Sigma | \Gamma \vdash \llbracket €t€ \rrbracket_\m{v} €:T€_1€=>T€ | €T€_y$ and
$\Sigma | \Gamma \vdash €t€_1€:T€_1 | €T€_y$. From the typing rule \textsc{T-Suspension}, we have
$\Sigma | \Gamma \vdash €v:T€_y | \bot$ and $\Sigma | \Gamma \vdash €t:T€_1€=>T€ | €T€_y$.
Now, it's easy to apply the typing rules to show that
$\Sigma | \Gamma \vdash \llbracket €t(t€_1€)€ \rrbracket_\m{v}€:T€ | €T€_y$.

\end{proof}

\begin{theorem}[Preservation]\label{thm:preservation}
If a term and the coroutine store are well-typed,
that is,
$\Sigma | \Gamma \vdash \texttt{t:T} | \texttt{T}_y$,
and
$\Sigma | \Gamma \vdash \mu$,
and if $\texttt{t} | \mu \rightarrow \texttt{t}' | \mu'$,
then there exists $\Sigma' \supseteq \Sigma$
such that
$\Sigma' | \Gamma \vdash \texttt{t}'\texttt{:T} | \texttt{T}_y$
and
$\Sigma' | \Gamma \vdash \mu'$.
\end{theorem}

\begin{proof}
We prove this by the induction on the typing derivations.
Cases \textsc{T-Unit}, \textsc{T-Abs}, \textsc{T-Coroutine},
\textsc{T-Empty} and \textsc{T-Instance} are straightforward,
since \texttt{t} is a value and does not reduce.

Note that we only consider reduction rules, ignoring
the rule \textsc{E-Context} and \textsc{E-Pause}. This is valid
because induction hypothesis makes the case \textsc{E-Context} trivial,
and \ref{lem:suspension} makes the case \textsc{E-Pause} trivial.

\proofcase{T-App}
$\texttt{t} = \texttt{t}_1\texttt{(t}_2\texttt{)}$

Consider the evaluation rule \textsc{E-AppAbs}.
Both \texttt{t}$_1$ and \texttt{t}$_2$ are values,
their yield type is $\texttt{T}_y = \bot$.
Moreover, then \texttt{t}$_1$ must have the form $\texttt{(x:T}_2\texttt{)=>t}_{11}$,
so by \textsc{E-AppAbs}
$\texttt{t'} = [\texttt{x} \mapsto \texttt{t}_2] \texttt{t}_{11}$.
From \ref{lem:inversion},
we know that
$\Sigma | \Gamma , \texttt{x:T}_2 \vdash \texttt{t}_{11}\texttt{:T}_1 | \bot$.
The claim about $\texttt{t}'$ follows from \ref{lem:substitution},
and $\mu' = \mu$.

\proofcase{T-Start}
$\texttt{t} = \texttt{start(t}_1\texttt{,t}_2\texttt{)}$

By the rule \textsc{E-Start},
$\texttt{t}_1$ and $\texttt{t}_2$ are values,
and $\texttt{t}' = i$ such that
$i \not\in dom(\mu)$.
Since $\mu$ is by assumption well-typed, it follows that $i \not\in dom(\Sigma)$.
But there exists
$\Sigma' = \Sigma , i \texttt{:T}_y \leftrightsquigarrow \texttt{T}_2 \supseteq \Sigma$
such that
$\Sigma' | \Gamma \vdash
i \texttt{:T}_y \leftrightsquigarrow \texttt{T}_2 | \texttt{T}_w$
and
$\Sigma' \vdash \mu'$.

\proofcase{T-Yield}
$\texttt{t} = \texttt{yield(t}_1\texttt{)}$

By the rule \textsc{E-Yield},
$\texttt{t}_1$ is a value,
and $\texttt{t}' = \llbracket \texttt{()} \rrbracket_{\texttt{v}}$.
From \textsc{T-Yield}, we know that
$\Sigma | \Gamma \vdash \texttt{t}_1 \texttt{:T}_y | \texttt{T}_y$.
From \textsc{T-Suspension},
we know that $\Sigma | \Gamma \vdash \texttt{t}' \texttt{:Unit} | \texttt{T}_y$.
We also know that $\mu' = \mu$, which proves the claim.

\proofcase{T-Snapshot}
$\texttt{t} = \texttt{snapshot(t}_1\texttt{)}$

This case is similar to \textsc{T-Yield},
but we rely on \textsc{E-Snapshot} for the transition,
and on \textsc{T-Instance} to type the resulting term $\texttt{t}' = i$.

\proofcase{T-Resume}
$\texttt{t} = \texttt{resume(t}_1\texttt{,t}_2\texttt{,t}_3\texttt{,t}_4\texttt{)}$

This case is similar to \textsc{T-Yield},
but we distinguish two cases --
that the value $\texttt{t}_1 = i$ is a terminated coroutine instance,
in which case we rely on \textsc{E-Resume2} for the transition,
and that $i$ is not terminated,
in which case we rely on \textsc{E-Resume1} for the transition.
We furthermore rely on \textsc{T-App} and \textsc{T-Unit}
to prove the typing relation on $\texttt{t}'$ in the former case,
and on \textsc{T-Resumption} in the latter.
In both of these cases,
we rely on \textsc{T-Empty}, \textsc{T-Suspension} and \textsc{T-Resumption} to
establish that $\Sigma \vdash \mu'$.

\proofcase{T-AppCor}
$\texttt{t} = \texttt{t}_1\texttt{(t}_2\texttt{)}$

By \textsc{E-AppCor},
we know
$\texttt{t}_1 = \texttt{(x:T}_2\texttt{)}
\overset{\texttt{T}_y}{\rightsquigarrow}
\texttt{t}_{11}$,
$\texttt{t}' = [\texttt{x} \mapsto \texttt{t}_2] \texttt{t}_{11}$,
and $\mu' = \mu$,
so the result follows from \ref{lem:substitution}.

\proofcase{T-Fix}
$\texttt{t} = \texttt{fix(t}_1\texttt{)}$

Trivial by \textsc{E-Fix} and \ref{lem:substitution}.

\proofcase{T-Resumption}
$\texttt{t} = \langle \texttt{t}_1\texttt{,v}_2\texttt{,v}_3\texttt{,v}_4 \rangle_i$

There are two subcases.
(1) The reduction rule is \textsc{E-Terminate}. we have
$\texttt{t}' = \texttt{t}_2\texttt{(t}_1\texttt{)}$.
The claim then follows from \textsc{T-App}, \textsc{T-Empty}
and \textsc{T-Resumption}.
(2) The reduction rule is \textsc{E-Capture}. The claim similarly follows.

\proofcase{T-Suspension}
$\texttt{t} = \llbracket \texttt{t}_1 \rrbracket_{\texttt{v}}$

Trivial, as a suspension cannot take a step.
\end{proof}

We are now ready to state another safety property
that follows directly from the preservation theorem.
We want to show that if the program was typed such that the yield type is $\bot$,
then the program will not yield a value outside
of a coroutine resumption.

\begin{cor}[Yield safety]\label{cor:yield-safety}
A well-typed user program \texttt{t}$_u$ never evaluates to a suspension term
$\llbracket \texttt{t} \rrbracket_{\texttt{v}}$.
\end{cor}

\begin{proof}
Assume that the program evaluates to the suspension term
$\llbracket \texttt{t} \rrbracket_{\texttt{v}}$.
Then, by \textsc{T-Suspension},
it must be that
$\Sigma | \varnothing \vdash \llbracket \texttt{t} \rrbracket_{\texttt{v}}
\texttt{:T} | \texttt{T}_y$.
Note that $\texttt{T}_y \neq \bot$,
since there is no non-empty term $\texttt{v}$ whose type is $\bot$,
and the empty term cannot appear in the evaluation.
Furthermore, by \ref{thm:preservation},
the original program \texttt{t}$_u$ must have
the same yield type $\texttt{T}_y \neq \bot$.
This is a contradiction, since such a program
would not be a well-typed user program, by \ref{def:well-typed-term}.
\end{proof}


\section{Formal Transformation of $\lambda_\leadsto$}
\label{sec:formal-transformation}

The formal transformation translates $\lambda_\leadsto$ programs
to programs in a simply typed lambda calculus extended with references
and restricted sums.
Intuitively, references are necessary because coroutine instances
are stateful entities -- we use references to store the evaluation state.
The key idea is to translate coroutines into functions
that accept a \emph{store function} as an argument,
as explained shortly.
We begin by introducing the target language.

\begin{definition}[Target language]\label{def:simply-typed-references}
The syntax of the target language is as follows:

\noindent
\begin{minipage}[t]{4.0 cm}
\small
$\ba[t]{l@{\hspace{-7mm}}r}
\texttt{t} ::= & \defmbox{terms:} \\
\quad \texttt{(x:T)=>t}
& \defmbox{abstraction}  \\
\quad \texttt{t(t)}
& \defmbox{application}  \\
\quad \texttt{x}
& \defmbox{variable}  \\
\quad \texttt{()}
& \defmbox{unit value}  \\
\quad \texttt{ref(t)}
& \defmbox{new reference}  \\
\quad \texttt{!t}
& \defmbox{dereferencing}  \\
\quad \texttt{t:=t}
& \defmbox{assignment}  \\
\ea$
\end{minipage}
\begin{minipage}[t]{5.7cm}
\small
$\ba[t]{l@{\hspace{-12mm}}r}
\quad
& \defmbox{}  \\
\quad \texttt{Ret(t)}
& \defmbox{tagged return}  \\
\quad \texttt{Yield(t)}
& \defmbox{tagged yield}  \\
\quad \texttt{Term}
& \defmbox{tagged ended}  \\
\quad \texttt{t match \{}
& \defmbox{pattern match} \\
\quad \texttt{case Ret(x) => t;} \\
\quad \texttt{case Yield(x) => t;} \\
\quad \texttt{case Term => t \}} \\
\ea$
\end{minipage}
\begin{minipage}[t]{4.5 cm}
\small
$\ba[t]{l@{\hspace{-8mm}}r}
\texttt{T} ::= & \defmbox{types:} \\
\biggap \texttt{T => T}
& \defmbox{function}  \\
\biggap \texttt{Unit}
& \defmbox{unit}  \\
\biggap \texttt{Ref[T]}
& \defmbox{reference}  \\
\biggap \texttt{Out[T,T]}
& \defmbox{output}  \\
\ea$
\end{minipage}

\vspace{-0.3cm}
\qed
\end{definition}

Informally, references of type \texttt{Ref[T]} are created with
the \texttt{ref} expression, assigned with \texttt{:=} and dereferenced with \texttt{!}.
Output values of type $\texttt{Out[T}_y\texttt{,T}_r\texttt{]}$ describe the
result of resuming a coroutine --
either a \texttt{Yield(x)}, where \texttt{x} has the type $\texttt{T}_y$,
indicating a yield;
or a \texttt{Ret(x)}, where \texttt{x} has the type $\texttt{T}_r$,
indicating a normal return;
or a \texttt{Term}, indicating a termination.
A pattern match reduces to the term in the respective case.
We do not show the exact typing rules and the operational semantics
for the target language,
since this was already treated in-depth \cite{Harper94asimplified,pierce02}.
We also skip runtime terms, as they are not used in the transformation.

\paratitle{Translation approach}
The transformation applies only to terms
that are lexically enclosed by a coroutine definition.
These terms are transformed into a continuation-passing style (CPS) --
the result of evaluating every term gets passed to a function that represents
the remainder of the enclosing coroutine (\emph{not the entire program}).
For a term of type \texttt{T},
this function takes the value of type \texttt{T},
and returns either a \texttt{Yield} value or a \texttt{Ret} value.

\begin{definition}\label{def:continuation-type}
The \emph{continuation type} is defined as
$
\varkappa\texttt{[T},\texttt{T}_y,\texttt{T}_r\texttt{]}
\triangleq
\texttt{T=>Out[T}_y,\texttt{T}_r\texttt{]}
$.
\end{definition}

Since the evaluation rules in Figure \ref{fig:evaluation-rules}
relied on the instance store $\mu$,
a coroutine instance must translate into a stateful entity.
Concretely, a coroutine instance will become a reference
that stores the continuation of the coroutine's execution,
typed $\texttt{Ref[Unit => Out[T}_y\texttt{, T}_r\texttt{]]}$.
To yield is to modify this reference.
To resume is to read it and run the continuation.

\begin{definition}\label{def:evaluation-state}
The \emph{evaluation state type} is defined as
$
\rho\texttt{[T}_y,\texttt{T}_r\texttt{]}
\triangleq
\texttt{Ref[}\varkappa\texttt{[Unit},\texttt{T}_y,\texttt{T}_r\texttt{]]}
$.
\end{definition}

A coroutine definition translates into a lambda
that takes two arguments.
One argument must obviously correspond
to the coroutine's argument.
The other argument must encode the runtime state
of the coroutine instance.
At first glance, it is tempting to model this state
with the coroutine instance reference,
which would make the coroutine type:

{
\centering
$
\texttt{Ref[}\varkappa\texttt{[Unit},\texttt{T}_y,\texttt{T}_r\texttt{]] => }
\texttt{T => }
\texttt{Out[T}_y\texttt{,T}_r\texttt{]}
$

}

However, such a type would not allow modeling stackful coroutines.
Recall that a coroutine can be either started with \texttt{start},
or called by another coroutine.
In the latter case, a \texttt{yield} inside a coroutine
must provide a continuation that captures not only the current coroutine,
but also the caller coroutine
(i.e. yielding captures the entire call stack).
But, lexically speaking, a coroutine
has no way of knowing what the continuation of its caller is.
A coroutine can only create a continuation for its own scope,
and pass that continuation fragment to its caller.
The caller can then recursively extend the continuation with its own fragment.
Once the bottom of the call stack is reached,
the continuation is stored into the reference.
Therefore, we need to abstract this with a separate \emph{store function}.

\begin{definition}\label{def:store-and-coroutine-type}
The \emph{store function type}
(i.e. a function that stores the continuation),
is defined as
$
\sigma\texttt{[T}_y,\texttt{T}_r\texttt{]}
\triangleq
\varkappa\texttt{[Unit},\texttt{T}_y,\texttt{T}_r\texttt{]=>Unit}
$.
The \emph{coroutine definition type}
(i.e. the coroutine's equivalent after transformation)
is
$
\gamma\texttt{[T}_1,\texttt{T}_y,\texttt{T}_r\texttt{]}
\triangleq
\sigma\texttt{[T}_y,\texttt{T}_r\texttt{]=>T}_1
\texttt{=>Out[T}_y\texttt{,T}_r\texttt{]}
$.
\end{definition}

\begin{definition}\label{def:target-language-sugar}
The terms
$\texttt{t}_1\texttt{;t}_2$,
$\texttt{val}\,\,\texttt{x:T=t}_1\texttt{;t}_2$
and \texttt{()=>t}
are syntactic sugar:
$
\texttt{t}_1;\texttt{t}_2
\triangleq
\texttt{((u:Unit)=>t}_2\texttt{)(t}_1\texttt{)}
$,
and
$
\texttt{val x:T=t}_1\texttt{;t}_2
\triangleq
\texttt{((x:T)=>t}_2\texttt{)(t}_1\texttt{)}
$,
and
$
\texttt{()=>t}
\triangleq
\texttt{(u:Unit)=>t}
$.
\end{definition}

The translation scheme is as follows.
A \texttt{start} term
becomes the creation of a reference $\rho\texttt{[T}_y,\texttt{T}_r\texttt{]}$.
A \texttt{resume} becomes a read from this reference,
followed by a call to a continuation lambda.
A \texttt{yield} calls the store function to store the current continuation lambda.
Finally, when a coroutine calls another coroutine,
a new store function is created,
which takes the continuation of the callee,
and chains it with the continuation of the callsite.

\vspace{0.1cm}
\noindent
\paratitle{Example}
Consider the \texttt{dup} coroutine,
which simply returns the sum of its arguments,
and is translated as follows:

{
\centering
$
\texttt{(x:Int)} \overset{\bot}{\rightsquigarrow} \texttt{x+x}
\quad
\rightarrow
\quad
\texttt{(s:} \sigma\texttt{[}\bot,\texttt{Int}\texttt{]} \texttt{)}
\texttt{=>}
\texttt{(x:Int)}
\texttt{=>}
\texttt{s(()=>Term);Ret(x+x)}
$

}

The body of each coroutine
undergoes a variant of the CPS transform
\cite{Sabry:1992:RPC:141478.141563,Reynolds1993}.
Our transform is particular in the sense
that each term translates not to function that consumes a continuation,
but to a function that takes a store function and a continuation.
In the example above,
the store function is invoked immediately before returning from the coroutine,
in order to update the instance state.

\vspace{0.1cm}
\noindent
\paratitle{Example}
The following coroutine yields the argument once,
and is translated as follows:

{
\centering
$
\texttt{(x:Int)} \overset{\texttt{Int}}{\rightsquigarrow} \texttt{yield(x)}
\quad
\rightarrow
\quad
\texttt{(s:} \sigma\texttt{[}\texttt{Int},\texttt{Unit}\texttt{]} \texttt{)}
\texttt{=>}
\texttt{Int}
\texttt{=>}
\texttt{(y')}
\texttt{(s)(q)}
$

$
\texttt{y'} =
\texttt{(s:}\sigma\texttt{[Int,Unit]}\texttt{)=>}
\texttt{(k:}\varkappa\texttt{[Unit},\texttt{Int},\texttt{Unit}\texttt{]}\texttt{)=>}
\texttt{s(k); Yield(x)}
$

$
\texttt{q} = \texttt{(y:Unit)=>s(()=>Term);Ret(y)}
$

}

The coroutine body is in continuation-passing style.
To see that this transformation is correct,
imagine that we passed the store function and an integer to the translated coroutine.
By evolving the term, we would eventually
pass the continuation to the store function \texttt{s},
and arrive at the \texttt{Yield} value.
We define the precise transformation relation on terms next.

\begin{definition}\label{def:type-sub-and-xi}
The \emph{type substitution} $\tau$ replaces the types of a term as follows:

$
\tau(\texttt{t})
\triangleq
\forall \texttt{T}_1,\texttt{T}_y,\texttt{T}_2
[
\texttt{(x:T}_1\texttt{)} \overset{\texttt{T}_y}{\rightsquigarrow} \texttt{T}_2
\mapsto
\gamma\texttt{[T}_1,\texttt{T}_y,\texttt{T}_2\texttt{]}
]
[
\texttt{T}_y \leftrightsquigarrow \texttt{T}_2
\mapsto
\rho\texttt{[T}_y,\texttt{T}_2\texttt{]}
]
\texttt{t}
$
\\
The notation $\xi$
is an abbreviation used to express transformed terms,
and is defined as:

$
\xi(\texttt{T}_y,\texttt{T}_r,\texttt{T},\texttt{t})
\triangleq
\tau(
\texttt{(s:}\sigma\texttt{[T}_y,\texttt{T}_r\texttt{]}\texttt{)=>}
\texttt{(k:}\varkappa\texttt{[T},\texttt{T}_y,\texttt{T}_r\texttt{]}\texttt{)=>}
\texttt{t}
)
$
\end{definition}

\begin{definition}[Transformation relation]
The \emph{transformation relation} is a five place relation
$\Gamma \vdash \texttt{T}_y | \texttt{T}_r | \texttt{t} \rightarrow \texttt{t}'$,
where $\Gamma$ is the typing context,
$\texttt{T}_y$ is the yield type of the current term,
$\texttt{T}_r$ is the return type of the enclosing coroutine,
$\texttt{t}$ is the term in the source language,
and $\texttt{t}'$ is the term in the target language.
This relation is inductively defined
according to Figures \ref{fig:appendix-transformation}
and \ref{fig:appendix-transformation-outside}.
\qed
\end{definition}

\begin{figure}[t]
  \small
  \centering
  \hspace*{-1.2cm}
  \input{transformation-rules.tex}
  \caption{
  Transformation of coroutines in $\lambda_\leadsto$
  to a simply typed lambda calculus with references
  }
  \label{fig:appendix-transformation}
\end{figure}

The rules fall into two groups.
The first is CPS-based and transforms terms inside coroutines (Fig. \ref{fig:appendix-transformation}).
The second group transforms terms outside of coroutines (Fig. \ref{fig:appendix-transformation-outside}).
All rules in the first group assume that we are inside a body of some coroutine,
which has the yield type $\texttt{T}_y$ and a return type $\texttt{T}_r$
(outside of a body of a coroutine, as we show shortly,
there is no need for a CPS transform).

\begin{wrapfigure}{L}{0.37\textwidth}
  \begin{minipage}{0.37\textwidth}
  \small
  \infrule[\textsc{X-Unit}]
  { \texttt{T}_y \neq \bot
  }
  { \Gamma \vdash \texttt{T}_y | \texttt{T}_r | \texttt{()} \rightarrow
    \xi(\texttt{T}_y, \texttt{T}_r, \texttt{Unit}, \texttt{k(}\texttt{()}\texttt{)})
  }
  \end{minipage}
\end{wrapfigure}

Consider the \textsc{X-Unit} rule,
which transforms \texttt{Unit} constants \emph{inside} coroutines.
A unit constant becomes a $\xi$-function
that takes a store function of type
$\sigma\texttt{[T}_y,\texttt{T}_r\texttt{]}$,
where $\texttt{T}_y$ and $\texttt{T}_r$
are the yield and return type of the surrounding context.
However, the transformed \texttt{Unit} constant
does not yield, and hence does not need to invoke the store function.
Instead, it just invokes the current continuation \texttt{k},
passing it the unit value.

\begin{wrapfigure}{R}{0.47\textwidth}
\begin{minipage}{0.47\textwidth}
\small
\infrule[\textsc{X-Yield}]
{ \texttt{T}_y \neq \bot \quad
  \Gamma \vdash \texttt{t} \texttt{:T} | \texttt{T} \quad
  \Gamma \vdash \texttt{T} | \texttt{T}_r | \texttt{t} \rightarrow \texttt{t}'
  \\*
  \texttt{p} =
  \texttt{t}'\texttt{(s)((x:T)=>s(k);Yield(x))}
}
{ \Gamma \vdash \texttt{T} | \texttt{T}_r |
  \texttt{yield(t)} \rightarrow
  \xi(
  \texttt{T}, \texttt{T}_r, \texttt{Unit},
  \texttt{p}
  )
}
\noindent
\end{minipage}
\end{wrapfigure}

On the other hand, the result of the \textsc{X-Yield} rule
must use the store function \texttt{s}
to store the continuation before yielding.
Given a term \texttt{t} that is recursively translated to $\texttt{t}'$,
the \texttt{yield(t)} term is translated to a $\xi$-function
that first evaluates $\texttt{t}'$
by passing it the store function and the continuation.
The continuation body stores the remainder of the continuation with \texttt{s},
and it itself reduces to \texttt{Yield(x)},
where \texttt{x} is the result of evaluating $\texttt{t}'$.

The \textsc{X-Var} rule simply calls the continuation by
passing the identifier \texttt{x} to the continuation.
The \textsc{X-App} rule assumes that the
subterm $\texttt{t}_1$ evaluates to a function $\texttt{t}_1'$,
and $\texttt{t}_2$ evaluates to $\texttt{t}_2'$.
Consider what happens after both $\texttt{t}_1'$ and $\texttt{t}_2'$
produce values $\texttt{x}_1$ and $\texttt{x}_2$, respectively.
The next step in the evaluation is to simply apply
$\texttt{x}_1$ to $\texttt{x}_2$.
However, the application $\texttt{x}_1\texttt{(x}_2\texttt{)}$
is a continuation for the term $\texttt{t}_2'$,
whose evaluation is itself a continuation for $\texttt{t}_1'$.
The transformation of $\texttt{t}_1\texttt{(t}_2\texttt{)}$
requires chaining these continuations together
as described in the \textsc{X-App} rule.

The \textsc{X-Coroutine} rule
relies on the assumption that the coroutine body \texttt{t}
translates to a function $\texttt{t}'$
under the extended typing context $\Gamma , \texttt{x:T}_1$.
The transformed term has the $\gamma$ type --
it takes a store function and an argument.
The store function is passed to $\texttt{t}'$,
along with a continuation that
stores a terminated continuation \texttt{()=>Term},
and wraps the result of $\texttt{t}'$
into a $\texttt{Ret}$ value.
The resulting $\gamma$ value is then passed
to the continuation of the definition site of the coroutine
(which is itself assumed to be inside another coroutine).

A coroutine created this way is, by rule \textsc{X-Start}, \emph{called}
by first allocating a reference \texttt{x} for the evaluation state
(term $\texttt{q}_1$),
then assigning the continuation of the coroutine
into the reference \texttt{x} (term $\texttt{q}_2$),
and finally, passing the reference \texttt{x}
to the continuation \texttt{k} (term $\texttt{p}$).
Note that the abbreviation $\psi$ is used to construct
the store function -- this is just a store function that assigns
the continuation to the reference.

\begin{definition}\label{def:helper-functions}
The \emph{output transformer} $\Phi$
is a function that,
for a given continuation value \texttt{k} of type
$\varkappa\texttt{[T}_r,\texttt{T}_y,\texttt{T}_q\texttt{]}$,
maps a value of type
$\texttt{Out[T}_y\texttt{,T}_r\texttt{]}$
to a value of type
$\texttt{Out[T}_y\texttt{,T}_q\texttt{]}$,
and is defined as follows:

{
\centering
\small
$
\Phi(\texttt{T}_y, \texttt{T}_r, \texttt{T}_q, \texttt{k})
$
\normalsize
$
\triangleq
$

\small
$
\texttt{(x:Out[T}_y\texttt{,T}_r\texttt{])=>}
\texttt{x match\{case Ret(x)=>k(x);case Yield(x)=>Yield(x);case Term=>Term\}}
$
\normalsize

}

\vspace{1mm}
The \emph{store function constructor} $\psi$
is a function that
maps a reference \texttt{x} of type
$\rho\texttt{[T}_y,\texttt{T}_r\texttt{]}$
to a store function of type
$\sigma\texttt{[T}_y,\texttt{T}_r\texttt{]}$,
and is defined as follows:

{
\centering
\small
$
\psi(\texttt{T}_y, \texttt{T}_r, \texttt{x})
$
\normalsize
$
\triangleq
$
\small
$
\texttt{(x:}\rho\texttt{[T}_y,\texttt{T}_r\texttt{]}\texttt{)=>}
\texttt{(k:}\varkappa\texttt{[Unit},\texttt{T}_y,\texttt{T}_r\texttt{])=>x:=k}
$
\normalsize

}

\qed
\end{definition}

We are now ready to take a look at the \textsc{X-AppCor} rule.
A transformed coroutine application
first constructs a mapping \texttt{f}
between output type of the callee coroutine
and the output type of the caller coroutine (term $\texttt{q}_1$),
whose type is defined as follows:
$
\varphi\texttt{[T}_y,\texttt{T}_r,\texttt{T}_q\texttt{]}
\triangleq
\texttt{Out[T}_y\texttt{,T}_r\texttt{]=>}
\texttt{Out[T}_y\texttt{,T}_q\texttt{]}
$.
This mapping is defined by the abbreviation $\Phi$
(defined in \ref{def:helper-functions}),
which just forwards \texttt{Yield} values.
However, if the callee returns a \texttt{Ret(x)} value,
then the wrapped value \texttt{x} is passed
to the continuation \texttt{k} of the callee
(which itself must return the correctly typed \texttt{Ret} value).

Next, the coroutine application must
create a new store function $\texttt{s}'$
with the appropriate type (term $\texttt{q}_2$).
This new store function
$\texttt{s}'$ passes a modified continuation to
the store function $\texttt{s}$ of the callee,
such that the modified continuation calls
the callee's continuation $\texttt{k}'$,
and then adapts the result using the mapping \texttt{f}.
Finally,
the transformed coroutine $\texttt{x}_1$
is invoked with the new store function $\texttt{s}'$,
and the result transformed using \texttt{f}
(term $\texttt{q}_3$).

Having seen \textsc{X-AppCor},
rules \textsc{X-Snapshot} and \textsc{X-Resume} should be self-explanatory.
The \texttt{snapshot} term
translates to the cloning the $\rho$ reference,
and the \texttt{resume} term
translates to dereferencing the continuation from the $\rho$ reference,
invoking it, and calling the correct handler, case-wise.

We deliberately left \textsc{X-Abs} as the last rule in our discussion.
Recall that the standard abstraction term effectively \emph{resets} the yield type
back to $\bot$ --
the standard abstraction term \emph{cannot} yield back to its callsite.
For this reason, the \textsc{X-Abs} rule
transforms the body of an abstraction term
under the assumption that the enclosing yield and return types are both $\bot$.
However, the fineprint present in each transformation rule
is that the $\texttt{T}_y$ \emph{must not} be $\bot$ --
the transformation rules shown so far cannot translate such a term!

The conclusion is that we need a second set of rules
that govern the transformation \emph{outside} of coroutines.
Naturally, these rules do not require a CPS transform,
and they leave most of the program as is --
they deal with coroutine definitions
and handling coroutine instances.
In particular, there is no analogue for the \textsc{X-AppCor}
and \textsc{X-Yield} rules --
a well-typed user program cannot yield or call coroutines
outside of a lexical scope of a coroutine.

\begin{figure}[t]
  \small
  \centering
  \hspace*{-1.2cm}
  \input{transformation-rules-outside.tex}
  \caption{
  Transformation of non-yielding terms in $\lambda_\leadsto$
  to a simply typed lambda calculus with references
  }
  \label{fig:appendix-transformation-outside}
\end{figure}

The rules in Fig. \ref{fig:appendix-transformation-outside},
whose name starts with \textsc{X-Free},
transform non-yielding terms.
They are analogous to the coroutine transformation rules.
When the yield and return types are $\bot$,
the transformation produces a normal term instead of a $\xi$ expression --
the logic inside each translated coroutine-related operation stays the same,
but the value is not being passed to a continuation.

To switch to the CPS transform from Fig. \ref{fig:appendix-transformation},
the transformation must apply the \textsc{X-FreeCoroutine} rule,
which transforms a coroutine definition.

{
\small
\infrule[\textsc{X-FreeCoroutine}]
{ \texttt{v} = \texttt{(x:T}_1\texttt{)}
  \overset{\texttt{T}_y}{\rightsquigarrow} \texttt{t}
  \quad
  \Sigma | \Gamma \vdash \texttt{v}
  \texttt{:T}_1\overset{\texttt{T}_y}{\rightsquigarrow}\texttt{T}_2 | \bot
  \quad
  \Gamma , \texttt{x:T}_1 \vdash
  \texttt{T}_y | \texttt{T}_2 | \texttt{t} \rightarrow \texttt{t}'
  \quad
  \texttt{T}_y \neq \bot
  \\*
}
{ \Gamma \vdash \bot | \bot |
  \texttt{v}
  \rightarrow
  \tau(
  \texttt{(s:}\sigma\texttt{[T}_y,\texttt{T}_2\texttt{])=>(x:T}_1\texttt{)=>}
  \texttt{t}'\texttt{(s)((x}_2\texttt{:T}_2\texttt{)=>}
  \texttt{s(()=>Term);Ret(x}_2\texttt{))}
  )
}
}

We read this as follows.
Consider a coroutine definition \texttt{v} whose body is the term \texttt{t},
and whose type is
$
\texttt{T}_1\overset{\texttt{T}_y}{\rightsquigarrow}\texttt{T}_2
$.
If the term \texttt{t} translates to a term $\texttt{t}'$,
then the coroutine definition translates to
a function that takes a store function \texttt{s}
and the argument \texttt{x},
and then invokes the transformed body $\texttt{t}'$
with the store function and the continuation.
This continuation stores the \texttt{Term}-returning function (i.e. terminates),
and returns the result $\texttt{Ret(x)}$.
From this, we can conclude that the transformed body $\texttt{t}'$
must be a $\xi$-function from \ref{def:type-sub-and-xi}.

Note, finally, that the only terms that are different
after the transformation are those that mention
elements of the $\lambda_\leadsto$ calculus syntax.
Any term that makes no mention of
$\lambda_\leadsto$-specific terms is left unchanged.
This shows that the transformation is both selective,
and that it does not require the recompilation of legacy programs
that use coroutines.


\section{Related Work}\label{sec:related}

We organize the related work on coroutines into two categories. We start with
the origins and previous formalization approaches, and then discuss related work
on continuations. We do not discuss related concepts like iterators or generators,
and interested readers can refer to our main work \cite{prokopec2018coroutines}.
Whereas coroutines have been studied extensively before,
the main novelty in our work is to augment coroutines with snapshots.
This allows for several new useful use-cases.

\textbf{Origins and formalizations.}
The idea of coroutines dates back to Erdwinn and Conway's work on
a tape-based Cobol compiler and its separability into modules
\cite{Conway:1963:DST:366663.366704}.
Although the original use-case is no longer relevant, other use-cases emerged.
Coroutines were investigated on numerous occasions,
and initially appeared in languages such as Modula-2 \cite{wirth1985programming},
Simula \cite{Birtwhistle:1979:SB:1096934},
and BCPL \cite{journals/spe/MoodyR80}.
A detailed classification of coroutines is given
by Moura and Ierusalimschy \cite{Moura:2009:RC:1462166.1462167},
along with a formalization of asymmetric coroutines
through an operational semantics.
Moura and Ierusalimschy observed that asymmetric first-class stackful coroutines
have an equal expressive power as one-shot continuations,
but did not investigate snapshots,
which make coroutines equivalent to full continuations.
Anton and Thiemann showed that it is possible
to automatically derive type systems
for symmetric and asymmetric coroutines
by converting their reduction semantics
into equivalent functional implementations,
and then applying existing type systems
for programs with continuations \cite{Anton2010}.
James and Sabry identified the input and output types of coroutines
\cite{yield-mainstream},
where the output type corresponds
to the \textit{yield} type described in this paper.
The input type ascribes the value passed to the coroutine when it is resumed.
As a design tradeoff, we chose not to have explicit input values in our model.
First, the input type increases the verbosity of the coroutine type,
which may have practical consequences.
Second, as shown in \cite{prokopec2018coroutines},
the input type can be simulated with the return type of another coroutine,
which yields a writable location, and returns its value when resumed.
Fischer et al. proposed a coroutine-based programming model
for the Java programming language,
along with the respective formal extension of Featherweight Java
\cite{Fischer:2007:TLS:1244381.1244403}.

\textbf{Transformation-based continuations.}
Continuations are closely related to coroutines,
and with the addition of \texttt{snapshot}
the two can express the same programs.
Scheme supports programming with continuations via the \texttt{call/cc} operator,
which has a similar role as \texttt{shift}
in shift-reset delimited continuations
\cite{Danvy:1990:AC:91556.91622,Asai:2017:SCT:3175493.3162069}.
In several different contexts,
it was shown that continuations subsume other control constructs
such as exception handling, backtracking, and coroutines.
Nonetheless, most programming languages do not support continuations.
It is somewhat difficult to provide an efficient implementation of continuations,
since the captured continuations must be callable more than once.
One approach is to transform the program
to continuation-passing style \cite{Sussman:1998:SIE:609145.609186}.
Scala's continuations \cite{Rompf:2009:IFP:1631687.1596596}
implement delimited shift-reset continuations with a CPS transform.
The downside of CPS is the risk of runtime stack overflows
in the absence tail-call optimization, as is the case of JVM.

Optimizing compilers tend to be tailored to the workloads that appear in practice.
For example, it was shown that optimizations such as
inlining, escape analysis, loop unrolling and devirtualization
make most collection programs run nearly optimally
\cite{Prokopec:2017:MCO:3136000.3136002,7092728,Prokopec:196627,Sujeeth:2013:CRC:2524984.2524988,prokopec11}.
However, abstraction overheads associated with coroutines are somewhat new,
and are not addressed by most compilers.
For this reason, compile-time transformations of coroutine-heavy workloads
typically produce slower programs compared to their runtime-based counterparts.
We postulate that targeted high-level JIT optimizations
could significantly narrow this gap.

\textbf{Runtime-based continuations.}
There were several attempts to provide runtime continuation support for the JVM,
ranging from Ovm implementations \cite{LAMP-CONF-2007-002}
based on \texttt{call/cc},
to JVM extensions \cite{Stadler:2009:LCJ:1596655.1596679},
based on the \texttt{capture} and \texttt{resume} primitives.
While runtime continuations are not delimited and can be made very efficient,
maintenance pressure and portability requirements
prevented these implementations from becoming a part of official JVM releases.
An alternative, less demanding approach relies only on stack introspection facilities
of the host runtime \cite{Pettyjohn:2005:CGS:1090189.1086393}.
There exists a program transformation that relies on exception-handling
to capture the stack \cite{Sekiguchi:2001:PIC:647332.722736}.
Here, before calling the continuation,
the saved state is used in method calls to rebuild the stack.
This works well for continuations,
where the stack must be copied anyway,
but may be too costly for coroutine \texttt{resume}.
Bruggeman et al. observed that many use cases
call the continuation only once and can avoid the copying overhead,
which lead to \textit{one-shot continuations}
\cite{Bruggeman:1996:RCP:231379.231395}.
One-shot continuations are akin to coroutines without snapshots.

\textbf{Domain-specific approaches.}
One of the early coroutine applications was data structure traversal.
Push-style traversal with \texttt{foreach} is easy,
but the caller must relinquish control,
and many applications cannot do this
(e.g. the same-fringe benchmark,
in which two trees are traversed pairwise simultaneously).
Java-style iterators with \texttt{next} and \texttt{hasNext}
are harder to implement than a \texttt{foreach} method,
and coroutines bridge this gap.

Iterators in CLU \cite{Liskov:1977:AMC:359763.359789} are essentially coroutines --
program sections with \texttt{yield} statements
that are converted into traversal objects.
C\# inherited this approach -- its iterator type \texttt{IEnumerator} exposes
\texttt{Current} and \texttt{MoveNext} methods.
Since enumerator methods are not first class entities,
it is somewhat harder to abstract suspendable code.
C\# enumerators are not stackful,
so iterator definitions must be implemented inside a single method.
Enumerators can be used for asynchronous programming,
but they require exposing \texttt{yield} in user code.
Therefore, separately from enumerators, C\# exposes async-await primitives.
Some newer languages such as Dart
similarly expose an async-await pair of primitives.

Async-Await in Scala \cite{scala-async}
is implemented using Scala's metaprogramming facilities.
Async-await programs can compose
by expressing asynchronous components as first-class \texttt{Future} objects.
The Async-Await model does not need to be stackful,
since separate modules can be expressed as separate futures.
However, reliance on futures and concurrency makes it hard
to use Async-Await generically.
For example, iterators implemented using futures
have considerable performance overheads
due to synchronization involved in creating future values.

There exist other domain-specific suspension models.
For example,
Erlang's €receive€ statement
effectively captures the program continuation
when awaiting for the inbound message \cite{programming-erlang96}.
A model similar to Scala Async was
devised to generate Rx's \texttt{Observable} values
\cite{EPFL-CONF-188383,Meijer:2012:YMD:2160718.2160735},
and the event stream composition
\cite{Prokopec:2017:EBB:3133850.3133865,Prokopec:2014:CAM:2637647.2637656}
in the reactor model
\cite{reactors-website,10.1007/978-3-319-64203-113,Prokopec:2015:ICE:2814228.2814245,Prokopec:2016:PSR:3001886.3001891},
as well as callbacks usages
in asynchronous programming models based on futures and flow-pools
\cite{SIP14,prokopec12flowpools,EPFL-REPORT-181098,6264/THESES,Schlatter:198208}
can be similarly simplified.
Cilk's spawn-sync model \cite{Leiserson1997} is similar to async-await,
and it is implemented as a full program transformation.
The Esterel language defines a \texttt{pause} statement
that pauses the execution, and continues it in the next event propagation cycle
\cite{Berry:1992:ESP:147276.147279}.
Behaviour trees \cite{6907656} are AI algorithms used to simulate agents --
they essentially behave as AST interpreters with yield statements.

\bibliography{coroutines-ecoop}

\end{document}